\documentclass[letterpaper, 10 pt, conference]{ieeeconf}  

\IEEEoverridecommandlockouts                              

\usepackage{amsmath,amssymb,amsfonts}
\usepackage{algorithm}
\usepackage{algpseudocode}
\usepackage{graphicx}
\usepackage{textcomp}
\usepackage{xcolor}
\usepackage{stfloats}
\usepackage{booktabs}
\usepackage{hyperref}
\usepackage{cite}
\let\labelindent\relax
\usepackage{enumitem}
\usepackage{arydshln}

\hypersetup{hypertex=true,
	colorlinks=true,
	linkcolor=blue,
	anchorcolor=blue,
	citecolor=blue}

\newtheorem{assumption}{Assumption}[section]
\newtheorem{definition}{Definition}[section]

\newtheorem{theorem}{Theorem}[section]

\newtheorem{lemma}{Lemma}
\newtheorem{remark}{Remark}

\newcommand\norm[1]{\left\lVert#1\right\rVert}

\title{\LARGE \bf{ 
		Finite Sample Analysis for a Class of Subspace Identification Methods}}

\author{Jiabao He, Ingvar Ziemann, Cristian R. Rojas and H\r{a}kan Hjalmarsson
	\thanks{Jiabao He, Cristian R. Rojas and H\r{a}kan Hjalmarsson are with the Division of Decision and Control Systems, School of Electrical Engineering and Computer Science, KTH Royal Institute of Technology, 100 44 Stockholm, Sweden.  (Emails: jiabaoh, crro, hjalmars@kth.se)}
	\thanks{Ingvar Ziemann is with the University of Pennsylvania, Philadelphia, PA 19104 USA. (Email: ingvarz@seas.upenn.edu)}
	\thanks{This work was supported by VINNOVA Competence Center AdBIOPRO, contract [2016-05181] and by the Swedish Research Council through the research environment NewLEADS (New Directions in Learning Dynamical Systems), contract [2016-06079], and contract 2019-04956.}
}

\begin{document}

	\maketitle
	\thispagestyle{empty}
	\pagestyle{empty}

	\begin{abstract}
		While subspace identification methods (SIMs) are appealing due to their simple parameterization for MIMO systems and robust numerical realizations, a comprehensive statistical analysis of SIMs remains an open problem, especially in the non-asymptotic regime. In this work, we provide a finite sample analysis for a class of SIMs, which reveals that the convergence rates for estimating Markov parameters and system matrices are  $\mathcal{O}(1/\sqrt{N})$, in line with classical asymptotic results. Based on the observation that the model format in classical SIMs becomes non-causal because of a projection step, we choose a parsimonious SIM that bypasses the projection step and strictly enforces a causal model to facilitate the analysis, where a bank of ARX models are estimated in parallel. Leveraging recent results from finite sample analysis of an individual ARX model, we obtain an overall error bound of an array of ARX models and proceed to derive error bounds for system matrices via robustness results for the singular value decomposition.
	\end{abstract}
	
	\IEEEpeerreviewmaketitle
	
	\section{Introduction} \label{sec1}
	
	Originating from the celebrated Ho-Kalman algorithm \cite{Ho1966effective}, subspace identification methods (SIMs) have proven extremely useful for estimating linear state-space models. Over the past 50 years, numerous efforts have been made to develop improved algorithms and gain a deeper understanding of the family of SIMs, exemplified by the successful narrative of closed-loop identification \cite{Qin2003closed,Jansson2003subspace,Ljung1996subspace,Chiuso2005consistency}. For a comprehensive overview of SIMs, we refer to \cite{Qin2006overview,van2013closed}. Despite their tremendous success both in theory and practice, several drawbacks have been recognized, including a lower accuracy compared to prediction error methods (PEMs) and an incomplete statistical analysis. There are some significant contributions to statistical properties of SIMs in the asymptotic regime. The consistency of open-loop and closed-loop SIMs is analyzed in \cite{Jansson1998consistency} and \cite{Chiuso2005consistency}, respectively, where the former suggests that persistence of excitation (PE) of the input signals is not sufficient for consistency, and stronger conditions are required in some cases. The asymptotic variance of SIMs is discussed in \cite{Gustafsson2002subspace,Jansson2000asymptotic,Bauer2005asymptotic,Chiuso2004asymptotic}. Furthermore, the asymptotic equivalence of some SIMs is discussed in \cite{Chiuso2007relation,Chiuso2007role,Chiuso2007some}. Regarding the optimality, although the canonical variate analysis (CVA) method is stated to be optimal when measured inputs are white \cite{Larimore1996statistical,Bauer2002some}, simulation studies indicate that it is not asymptotically efficient, as it does not reach Cram\'er-Rao lower bound (CRLB) \cite{Chiuso2007role}. In short, SIMs are generally consistent, however, the question of whether there are SIMs that are asymptotically efficient remains unresolved. Besides, although some of the methods are asymptotically equivalent, their performance differs in finite sample setups. Meanwhile, it is difficult to capture their transient behaviors using the asymptotic theory. Therefore, a complete statistical analysis and comparison among various SIMs, as well as the pursuit of an efficient SIM are still open problems.
	
	There has been a recent resurgence of interest in identifying state-space models for dynamic systems, where the focus is on the non-asymptotic regime. 
	Finite sample analysis in the field of system identification was pioneered by \cite{Campi2002finite,Weyer1999finite}, where the performance of PEMs were analyzed. Over the last few years, a series of papers have revisited this topic and introduced many promising developments on fully observed systems \cite{Sarkar2019near,Simchowitz2018learning,Jedra2022finite} and partially observed systems \cite{Tsiamis2019finite,Oymak2019non,Sarkar2021finite,Lale2021finite}. For a broader overview of these results, we refer to \cite{Tsiamis2023statistical,Ziemann2023tutorial}. As pointed out in \cite{Tsiamis2019finite}, finite sample analysis has been a standard tool for comparing algorithms in the non-asymptotic regime. It is expected that such analysis of SIMs will not only provide a detailed qualitative characterization of learning complexity and error bounds, but also bring more insights into the comparison of different SIMs regarding optimality, the selection of past and future horizons, and the design of controllers \cite{Tsiamis2023statistical}. Broadly speaking, an analysis in the non-asymptotic regime brings a wider understanding of SIMs.
	
	Despite the auspicious future, the path of finite sample analysis for SIMs proves to be challenging. Usually, multi-step statistical operations are involved in SIMs, including pre-estimation, projection, weighted singular value decomposition (SVD) and maximum likelihood (ML) estimation. While these steps enhance performance, they simultaneously complicate the model and pose challenges for statistical analysis. To the best of our knowledge, except for a finite sample analysis of the Ho-Kalman algorithm \cite{Oymak2019non}, a stochastic SIM in the absence of input signals \cite{Tsiamis2019finite} and an individual ARX model \cite{Lale2021finite,Ziemann2023tutorial}, a finite sample analysis of classical SIMs in the presence of inputs is still unavailable. Compared with SIMs without taking account of measured inputs, the inclusion of inputs results in a higher complexity due to the presence of an unknown block matrix with a lower-triangular Toeplitz structure. This matrix is responsible for recording the impact of 'future' input on 'future' output. To manage this complexity, classical SIMs choose to remove this matrix via a projection step. Although this method is computationally efficient, it makes the model format non-causal anymore and poses challenges for statistical analysis. Therefore, we propose to use a parallel and parsimonious SIM, namely PARSIM \cite{Qin2005novel}, which bypasses the projection step and strictly enforces a causal model to facilitate the analysis. Except the projection step, PARSIM aligns with the unified framework of the family of SIMs in all other aspects, opting for PARSIM will not constrain our comprehension of the full panorama of SIMs.
	
	The main contributions of this paper are:
	\begin{enumerate}
		\item We leverage recent results from finite sample analysis of an individual ARX model to obtain an error bound for the bank of ARX models featured in PARSIM. Our analysis reveals that the convergence rates for estimating Markov parameters and system matrices are  $\mathcal{O}(1/\sqrt{N})$ even in the presence of inputs, which is in line with classical asymptotics.
		\item Our method can be extended to the class of SIMs that estimate an array of ARX models, for instance, the included PARSIM and SIMs based on predictor identification (PBSID). Therefore, it paves the way for comprehensively understanding the broader landscape of SIMs.
	\end{enumerate}
	
	The disposition of the paper is as follows: A short review of SIMs with the focus on PARSIM is given in Section \ref{sec2}, and the problem as well as the roadmap ahead to analyze its finite sample behavior are described at the end of Section \ref{sec2}. The finite sample analysis of an individual ARX model is summarized in Section \ref{sec3}. An overall error bound for a bank of ARX models and the resulting error bounds of the system matrices are given in Section \ref{sec4}. Finally, some discussions and future work are provided in Section \ref{sec5}.
	
	\textit{Notations:} For a matrix $X$ with appropriate dimensions, $X^\top$, $X^{-1}$, $X^{1/2}$, $X^\dagger$, $\lVert X \rVert$, ${\rm{det}}(X)$, $\rho(X)$, $\lambda_{\rm{max}}(X)$, $\lambda_{\rm{min}}(X)$ and $\sigma_{n}(X)$ denote its transpose, inverse, square root, Moore$\mbox{-}$Penrose pseudo-inverse, spectral norm, determinant, spectral radius, maximum eigenvalue, minimum eigenvalue and $n-$th largest singular value, respectively. $X \succ(\succcurlyeq)$ $0$ means that $X$ is positive (semi) definite. The matrices $I$ and $0$ are the identity and zero matrices with compatible dimensions. The multivariate normal distribution with mean $\mu$ and covariance $\Sigma$ is denoted as $\mathcal{N}(\mu,\Sigma)$. The notation ${\mathbb{E}{x}}$ is the expectation of a random vector $x$, and $\mathbb{P}(\mathcal{E})$ is the probability of the event $\mathcal{E}$. $\mathcal{E}^c$ is the complementary event of $\mathcal{E}$, and $\mathcal{E}_{1}\cup\mathcal{E}_{2}$ and $\mathcal{E}_{1}\cap\mathcal{E}_{2}$ are the union and intersection of events $\mathcal{E}_{1}$ and $\mathcal{E}_{2}$, respectively.
	
	\section{Preliminaries} \label{sec2}
	
	\subsection{Model and Assumptions} \label{Sct2.1}
	
	Consider the following discrete-time linear time-invariant (LTI) system on innovations form:
	\begin{subequations} \label{E1}
		\begin{align}
			x_{k + 1} &= Ax_{k}  + Bu_{k} + Ke_{k}, \label{E1a}\\
			y_{k} &= Cx_{k} +e_{k}, \label{E1b}		
		\end{align}
	\end{subequations}
	where $x_{t}\in \mathbb{R}^{n_x}$, $u_{t}\in \mathbb{R}^{n_u}$, $y_{t}\in \mathbb{R}^{n_y}$ and $e_{t}\in \mathbb{R}^{n_y}$ are the state, input, output and innovations, respectively. For brevity of notation, we assume that the initial time starts at $k=1$, and the initial state is $x_1=0$. It has been widely recognized that under mild conditions, the above innovations model describes the same input-output trajectories with identical statistics as a standard state-space model \cite{Ziemann2023tutorial,Qin2006overview}. Thus, without loss of generality, we study the innovations model. We make the following assumptions which are commonly used in SIMs:
	\begin{assumption} \label{Asp1}
		\begin{enumerate} 
			\item The spectral radius of $A$ and $A-KC$ satisfies $\rho({A}) \leq 1$ and $\rho({A-KC}) < 1$.
			\item The system is minimal, i.e., $(A,[B,K])$ is controllable and $(A,C)$ is observable, and the system order $n_x$ is known to the user.
			\item The innovations $\{e_k\}$ consists of independent and identically distributed (i.i.d.) Gaussian random variables, i.e., $e_{k} \sim \mathcal{N}(0,\sigma_e^2I)$.\footnote{We believe that our results can be extended to more general setups, for instance, sub-Gaussians.}
			\item The input sequence $\{u_k\}$ consists also of i.i.d. Gaussian random variables, i.e., $u_{k} \sim \mathcal{N}(0,\sigma_u^2I)$. Moreover, it is assumed independent of $\{e_k\}$.
		\end{enumerate}
	\end{assumption}
	
	\subsection{A Recap of PARSIM} \label{Sct2.2}
	
	Here we provide a short review of SIMs, with a focus on PARSIM. An extended state-space model for \eqref{E1} can be derived as \cite{Qin2005novel}
	\begin{subequations} \label{E2}
		\begin{align}
			Y_f &= \Gamma_fX_k + G_fU_f + H_fE_f, \label{E2a}\\
			Y_p &= \Gamma_pX_{k-p} + G_pU_p + H_pE_p, \label{E2b}		
		\end{align}
	\end{subequations}
	where $f$ and $p$ denote future and past horizons chosen by the user, respectively. The extended observability matrix is
	\begin{equation} \label{E3}
		\Gamma_f = \begin{bmatrix}
			{{C^{\top}}}&{{{\left( {CA} \right)}^{\top}}}& \cdots &{{{\left({C{A^{f - 1}}} \right)}^{\top}}}
		\end{bmatrix}^{\top},
	\end{equation}
	and $G_f$ with $H_f$ are lower-triangular Toeplitz matrices of Markov parameters with respect to the input and innovations, 
	\vspace{-3mm}
	\begin{subequations} \label{E4}
		\begin{align}
			G_{f} &= \begin{bmatrix}
				{0}&{0}& \cdots &0\\
				{CB}&0& \cdots &0\\
				\vdots & \vdots & \ddots & \vdots \\
				{C{A^{f-2}}B}&{C{A^{f-3}}B}& \cdots &0
			\end{bmatrix}, \label{E4a}\\
			H_{f} &= \begin{bmatrix}
				{I}&{0}& \cdots &0\\
				{CK}&I& \cdots &0\\
				\vdots & \vdots & \ddots & \vdots \\
				{C{A^{f-2}}K}&{C{A^{f-3}}K}& \cdots &I
			\end{bmatrix}. \label{E4b}
		\end{align}
	\end{subequations}
	Past and future inputs are collected in the Hankel matrices
	\begin{subequations} \label{E5}
		\begin{align}
			U_{p} &= \begin{bmatrix}
				{{u_{k-p}}}&{{u_{k-p+1}}}& \cdots &{{u_{k-p+N-1}}}\\
				{{u_{k-p+1}}}&{{u_{k-p+2}}}& \cdots &{{u_{k-p+N}}}\\
				\vdots & \vdots & \ddots & \vdots \\
				{{u_{k-1}}}&{{u_{k}}}& \cdots &{{u_{k+N-2}}} \end{bmatrix}, 	\label{E5a} \\
			U_{f} &= \begin{bmatrix}
				{{u_k}}&{{u_{k+1}}}& \cdots &{{u_{k+N-1}}}\\
				{{u_{k+1}}}&{{u_{k+2}}}& \cdots &{{u_{k+N}}}\\
				\vdots & \vdots & \ddots & \vdots \\
				{{u_{k+f-1}}}&{{u_{k+f}}}& \cdots &{{u_{k+f+N-2}}} \end{bmatrix}.  \label{E5c}
		\end{align}
	\end{subequations}
	Similar definitions are given for matrices $\Gamma_p$, $G_p$, $H_p$, $Y_p$, $Y_f$, $E_p$ and $E_f$ \cite{Qin2005novel}. Usually, we take $k=p+1$, and then the total number of samples is denoted as $\bar N = p+f+N-1$. The state sequences are defined as
	\begin{subequations} \label{E6}
		\begin{align}
			X_{k} &= \begin{bmatrix}
				{{x_k}}&{{x_{k+1}}}& \cdots &{{x_{k+N-1}}} \end{bmatrix},\label{E6a}\\
			X_{k-p} &= \begin{bmatrix}
				{{x_{k-p}}}&{{x_{k-p+1}}}& \cdots &{{x_{k-p+N-1}}} \end{bmatrix}.\label{E6b}
		\end{align}
	\end{subequations}
	Furthermore, by replacing $e_k$ with $y_k - Cx_k$ in \eqref{E1a} and iterating the equation, we obtain the relation
	\begin{equation} \label{E7}
		X_k = L_pZ_p + \bar A^pX_{k-p},
	\end{equation}
	where $A_c = A-KC$, $Z_p = \begin{bmatrix}
		Y_p^{\top}&U_p^{\top}
	\end{bmatrix}^{\top}$, and $L_p$ is the extended controllability matrix defined as
	\begin{equation} \label{E8}
		L_p = \begin{bmatrix}
			A_c^{p-1}K& \cdots &{A_c K}& K &{A_c^{p-1}B}&\cdots&{A_c B}&B
		\end{bmatrix}.
	\end{equation}
	After substituting \eqref{E7} into \eqref{E2a}, we have 
	\begin{equation} \label{E9}
		Y_f = \Gamma_fL_pZ_p + G_fU_f + H_fE_f + \Gamma_f\bar A^pX_{k-p}.
	\end{equation}
	Most SIMs use \eqref{E9} to first estimate either the extended observability matrix $\Gamma_f$ or system state $X_k$, and then obtain a realization of the system matrices up to a similarity transformation. To illustrate, since $G_f$ is a lower-triangular Toeplitz matrix recording the effect of $U_f$ on $Y_f$, and it is difficult to preserve such structure with least-squares, what classical SIMs do is to eliminate this term by projecting out $U_f$ as
	\begin{equation} \label{E10}
		Y_f\Pi_{U_f}^{\perp} = \Gamma_fL_pZ_p\Pi_{U_f}^{\perp} + H_fE_f\Pi_{U_f}^{\perp} + \Gamma_f\bar A^pX_{k-p}\Pi_{U_f}^{\perp},
	\end{equation}
	where $\Pi_{U_f}^{\perp} = I - U_f^\top(U_fU_f^\top)^{-1}U_f$. Then the above equation is simplified based on the following observations: When $p$ is sufficiently large, we have $\bar A^p \approx 0$. Also, as $U_f$ is uncorrelated with $E_f$, we have $E_f\Pi_{U_f}^{\perp} \approx E_f$. Furthermore, $E_f$ is uncorrelated with $Z_p$, i.e., $\frac{1}{N}E_fZ_p^\top \approx 0$. By multiplying $Z_p^\top$ on both sides of \eqref{E10} we have
	\begin{equation} \label{E11}
		Y_f\Pi_{U_f}^{\perp}Z_p^\top \approx \Gamma_fL_pZ_p\Pi_{U_f}^{\perp}Z_p^\top.
	\end{equation}
	Then the range space of the extended observability matrix $\Gamma_f$ can be estimated using
	\begin{equation} \label{E12}
		\widehat {\Gamma_fL_p} = Y_f\Pi_{U_f}^{\perp}Z_p^\top(Z_p\Pi_{U_f}^{\perp}Z_p^\top)^{-1}.
	\end{equation}
	To recover the extended observability matrix $\Gamma_f$ or the state $X_k$, weighted SVD is often used, i.e., 
	\begin{equation} \label{E13}
		W_1\widehat {\Gamma_fL_p}W_2 = \hat U\hat \Lambda \hat V^\top \approx \hat U_{1}\hat \Lambda_{1}\hat V_{1}^\top, 
	\end{equation}
	where $\hat \Lambda_{1}$ contains the $n_x$ largest singular values. In this way, a balanced realization of $\hat \Gamma_f$ or $\hat L_p$ is
	\begin{subequations} \label{E14}
		\begin{align}
			\hat \Gamma_f &= W_1^{-1}\hat U_{1}{\hat \Lambda}_{1}^{1/2},\label{E14a}\\
			\hat L_p &= {\hat \Lambda}_{1}^{1/2}\hat V_{1}^\top W_2^{-1}.\label{E14b}
		\end{align}
	\end{subequations}
	Different choices of weighting matrices $W_1$ and $W_2$ lead to distinct classical SIMs \cite{Van1995unifying,Qin2006overview}. As pointed out in \cite{Qin2005novel}, one of the main issues of those classical SIMs is that the model is not causal anymore due to the absence of $G_f$ in \eqref{E10}. As a result, the estimated parameters have inflated variance due to the existence of unnecessary and extra terms. Furthermore, this poses a challenge in analyzing statistical properties, which we will see later in detail. 
	
	To enforce causal models, a parallel and parsimonious SIM,  PARSIM, is proposed in \cite{Qin2005novel}. Instead of doing the projection step in \eqref{E10} once, PARSIM zooms into each row of \eqref{E9} and equivalently performs $f$ ordinary least-squares (OLS) to estimate a bank of ARX models. To illustrate this, the extended state-space model \eqref{E9} can be partitioned row-wise as 
	\begin{equation} \label{E15}
		Y_{fi} = \Gamma_{fi}L_pZ_p + G_{fi}U_i + H_{fi}E_i + \Gamma_{fi}\bar A^pX_{k-p},
	\end{equation}
	where $ i = 1,2,...f$, and
	\begin{equation} \label{E15a}
		\nonumber
		\begin{split}
			\Gamma_{fi} &= CA^{i-1}\in \mathbb{R}^{n_y\times n_x},\\
			Y_{fi} &= \begin{bmatrix}
				{{y_{k+i-1}}}&{y_{k+i}}& \cdots &{y_{k+N+i-2}} \end{bmatrix}\in \mathbb{R}^{n_y\times N},\\
			U_{fi} &= \begin{bmatrix}
				{{u_{k+i-1}}}&{u_{k+i}}& \cdots &{u_{k+N+i-2}} \end{bmatrix}\in \mathbb{R}^{n_u\times N},\\
			U_i &= \begin{bmatrix}
				{{U_{f1}^{\top}}}&{{U_{f2}^{\top}}}& \cdots &{{U_{fi}^{\top}}}
			\end{bmatrix}^{\top}\in \mathbb{R}^{in_u\times N}, \\
			E_{fi} &= \begin{bmatrix}
				{{e_{k+i-1}}}&{e_{k+i}}& \cdots &{e_{k+N+i-2}} \end{bmatrix}\in \mathbb{R}^{n_y\times N},\\
			E_i &= \begin{bmatrix}
				{{E_{f1}^{\top}}}&{{E_{f2}^{\top}}}& \cdots &{{E_{fi}^{\top}}}
			\end{bmatrix}^{\top}\in \mathbb{R}^{in_y\times N}, \\
			G_{fi} &= \begin{bmatrix}
				CA^{i-2}B&  \cdots &CB&0
			\end{bmatrix}  \\
			&\overset{\Delta}{=}\begin{bmatrix}
				G_{i-1}&  \cdots &G_{1}&G_{0}
			\end{bmatrix}\in \mathbb{R}^{n_y\times in_u}, \\
			H_{fi} &= \begin{bmatrix}
				CA^{i-2}K&  \cdots &CK&I
			\end{bmatrix} \\
			&\overset{\Delta}{=} \begin{bmatrix}
				H_{i-1}&  \cdots &H_{1}&H_{0}
			\end{bmatrix}\in \mathbb{R}^{n_y\times in_y}.
		\end{split}
	\end{equation}
	Then PARSIM uses OLS to estimate each $\Gamma_{fi}L_p$ and $G_{fi}$ simultaneously from the causal model \eqref{E15}:
	\begin{equation} \label{E16}
		\hat \theta_i \triangleq \begin{bmatrix}
			\widehat {\Gamma_{fi}L_p}& \hat G_{fi}
		\end{bmatrix} = Y_{fi} \begin{bmatrix}
			Z_p\\ U_i\end{bmatrix}^\dagger.
	\end{equation}
	At last, the whole estimate of $\Gamma_{f}L_p$ is obtained by stacking the $f$ estimates together as
	\begin{equation} \label{E17}
		\widehat {\Gamma_{f}L_p} = \begin{bmatrix}
			\widehat {\Gamma_{f1}L_p}\\ \widehat {\Gamma_{f2}L_p} \\ \vdots \\\widehat {\Gamma_{ff}L_p}\end{bmatrix}.
	\end{equation}
	Comparing with classical SIMs that only estimate $\Gamma_{f}L_p$ in \eqref{E12}, PARSIM estimates $\Gamma_{f}L_p$ and Markov parameters $G_{fi}$ simultaneously. In this way, the lower-triangular Toeplitz structure of $G_f$ is preserved and the causality is strictly enforced. Furthermore, it has been shown that the above algorithm gives a smaller variance of $\widehat {\Gamma_{f}L_p}$ than classical SIMs in the asymptotic regime. For the subsequent realization step, PARSIM goes back to the weighted SVD step \eqref{E13}. 
	
	In this paper, we consider a simplified implementation of the original PARSIM \cite{Qin2005novel}, i.e., we choose the weighting matrices $W_1 = I$ and $W_2 = I$, and the system matrices are obtained in the following way:
	\begin{subequations} \label{E18}
		\begin{align}
			\hat C &= \hat \Gamma_f(1:n_y,:), \label{E18a}\\
			\hat A &= (\underline{I}\hat \Gamma_{f})^\dagger(\bar{I}\hat \Gamma_{f}),\label{E18b}\\
			\hat K &= \hat L_{p}(:,(p-1)n_y+1:pn_y), \label{E18c}\\
			\hat B &= \hat L_{p}(:,(2p-1)n_y+1:2pn_y), \label{E18d}
		\end{align}
	\end{subequations}
	where $\underline{I} = \begin{bmatrix}I&0\end{bmatrix}$, $\bar{I} = \begin{bmatrix}0&I\end{bmatrix}$ and the indexing of matrices follows MATLAB syntax.
	
	\subsection{Problem and Roadmap Ahead} \label{Sct2.3}
	
	Now we define the problem explicitly and sketch the path ahead to its solution. Under Assumption \ref{Asp1}, given a finite number $\bar N$ of input-output samples and horizons $f$ and $p$, we aim to provide error bounds with high probability for the realization \eqref{E18}. To be specific, with probability at least $1-\delta$, we wish to establish the following error bounds explicitly \cite{Tsiamis2019finite}:
	\begin{subequations} \label{E19}
		\begin{align}
			&\lVert \hat A - T^{-1}AT\rVert \leq \epsilon_A(\delta,N),
			\lVert \hat C - CT\rVert \leq \epsilon_C(\delta,N),\label{E19b}\\
			&\lVert \hat K - T^{-1}K\rVert \leq \epsilon_K(\delta,N), 
			\lVert \hat B - T^{-1}B\rVert \leq \epsilon_B(\delta,N), \label{E19d}
		\end{align}
	\end{subequations}
	where $T$ is a non-singular matrix. 
	
	\begin{remark} \label{Rmk0}
		It is only possible to obtain the system matrices up to a similarity transformation due to the non-uniqueness of the realization\cite{Oymak2019non}.
	\end{remark}
	
	We arrive at this result by a two-step procedure. In Step 1, we derive an error bound for $\hat \theta_i$ in \eqref{E16} for every ARX model. In other words, we define the event 
	\begin{equation} \label{E20}
		\mathcal{E}_{i} \triangleq \left\{\norm{\hat \theta_i - \theta_i} \leq \epsilon_{\theta_i}(\delta,N)\right\},
	\end{equation}
	and require that $\mathbb{P}(\mathcal{E}_{i}^c) \leq \delta/f$ for $i=1,2,...,f$. In Step 2, we first utilize a norm inequality between the block matrix $\widehat {\Gamma_{f}L_p}$ and its sub-matrices $\widehat {\Gamma_{fi}L_p}$ to obtain the total error bound of $\widehat {\Gamma_{f}L_p}$. This essentially requires that the intersection of $f$ events has probability $\mathbb{P}(\mathcal{E}_{1}\cap\mathcal{E}_{2}\cap \cdots \cap\mathcal{E}_{f}) \geq 1-\delta$, which is guaranteed due to Bonferroni's inequality: the probability $\mathbb{P}(\mathcal{E}_{1}^c\cup\mathcal{E}_{2}^c\cup \cdots \cup\mathcal{E}_{f}^c) \leq \sum_{i=1}^{f} \mathbb{P}(\mathcal{E}_{i}^c) \leq \delta$, in Step 1. At last, using recent results from SVD robustness \cite{Oymak2019non,Tsiamis2019finite}, error bounds in \eqref{E19} are obtained.
	
	\section{Finite Sample Analysis of each ARX Model} \label{sec3}
	
	Following our roadmap, we formalize Step 1 above in this section. 
	We emphasize that the results presented in this section apply to $i=1,2,...,f$, with the acknowledgment of their reliance on the specific value of $i$. This dependency is underscored through the use of the subscript $i$. First, we partition the following matrices column-wise:
	\begin{subequations} \label{E21}
		\begin{align}
			\begin{bmatrix}
				Z_p\\ \hdashline[2pt/2pt] U_i\end{bmatrix}&= \begin{bmatrix}
				y_{p}(1)&y_{p}(2)&\cdots&y_{p}(N) \\
				u_{p}(1)&u_{p}(2)&\cdots&u_{p}(N) \\
				\hdashline[2pt/2pt]
				u_{i}(1)&u_{i}(2)&\cdots&u_{i}(N) \end{bmatrix}, \\
			E_i&= \begin{bmatrix}
				e_{i}(1)&e_{i}(2)&\cdots&e_{i}(N)\end{bmatrix},
		\end{align}
	\end{subequations}
	where $u_{p}(k) = \begin{bmatrix}
		{{u_{k}^\top}}&{{u_{k+1}^\top}}&\cdots&{{u_{k+p-1}^\top}}\end{bmatrix}^\top$ and $u_{i}(k) = \begin{bmatrix}{{u_{k+p}^\top}}&{{u_{k+p+1}^\top}}&\cdots&{{u_{k+p+i-1}^\top}} \end{bmatrix}^\top$ are past input and future input, respectively, and similar definitions apply to $y_{p}(k)$ and $e_{i}(k)$. For brevity, we define a covariate
	\begin{equation} \label{E22}
		z_{p,i}(k) = \begin{bmatrix}
			y_p(k)\\ u_p(k)\\ u_i(k)\end{bmatrix},
	\end{equation} 
	so the error of the OLS estimate \eqref{E16} can be written as 
	\begin{equation} \label{E23}
		\begin{split}
			{\tilde \theta}_{i} \overset{\Delta}{=} & {\hat \theta}_{i} - \theta_{i}= H_{fi}E_i\begin{bmatrix}
				Z_p\\ U_i\end{bmatrix}^\dagger + \Gamma_{fi}\bar A^pX_{k-p} \begin{bmatrix}
				Z_p\\ U_i\end{bmatrix}^\dagger \\
			= &\underbrace{H_{fi}\sum_{k=1}^{N}\frac{1}{N}{e_i(k)z_{p,i}^\top(k)} \left(\sum_{k=1}^{N}\frac{1}{N}{z_{p,i}(k)z_{p,i}^\top(k)}\right)^{-1}}_{{\text{stochastic error}} \ {\tilde \theta}_{i}^S} + \\
			&\underbrace{\Gamma_{fi}\bar A^p\sum_{k=1}^{N}\frac{1}{N}{x_{k}z_{p,i}^\top(k)} \left(\sum_{k=1}^{N}\frac{1}{N}{z_{p,i}(k)z_{p,i}^\top(k)}\right)^{-1}}_{{\text{truncation bias} \ {\tilde \theta}_{i}^B}}.
		\end{split}
	\end{equation}
	There are two types of errors, namely, the stochastic error ${\tilde \theta}_{i}^S$ and truncation bias ${\tilde \theta}_{i}^B$. The key observation is that the future innovations $e_i(k)$ are independent of the covariate $z_{p,i}(l)$ for all $l<k$, due to the fact that $z_{p,i}(l)$ consists of the past output, past input and future input. This provides a martingale structure, which is convenient to analyze. By contrast, if we revisit the classical SIMs, the stochastic error for the estimate \eqref{E12} will be $E_f\Pi_{U_f}^{\perp}Z_p^\top(Z_p\Pi_{U_f}^{\perp}Z_p^\top)^{-1}$. Due to the projection matrix $\Pi_{U_f}^{\perp}$, the columns of $E_f$ and $Z_p$ are mixed together, making the above term non-causal, resulting in the loss of the martingale structure. We believe that this is one of the main barriers preventing a finite sample analysis for classical SIMs, which is also the reason why we choose PARSIM that bypasses the projection step. Before proceeding further, we have the following definitions regarding the covariance and empirical covariance of the covariant $z_{p,i}(k)$:
	\begin{subequations} \label{E24}
		\begin{align}
			{\Sigma}_{p,i}(k) &\triangleq \mathbb{E}z_{p,i}(k)z_{p,i}^\top(k),\\
			{\hat \Sigma}_{p,i}(N) &\triangleq \frac{1}{N}\sum_{k=1}^{N}z_{p,i}(k)z_{p,i}^\top(k).
		\end{align}
	\end{subequations}
	For simplicity, with a slight abuse of notation, we use ${\Sigma}_{i,k}={\Sigma}_{p,i}(k)$ and ${\hat \Sigma}_{i,N}={\hat \Sigma}_{p,i}(N)$, where the dependency of covariance on the past horizon $p$ is concealed. 
	Also, the covariance of the state $x_k$ is defined similarly as
	\begin{equation} \label{E25}
		{\Sigma}_{x,k} \triangleq \mathbb{E}x_{k}x_{k}^\top.
	\end{equation}
	In this way, the stochastic error ${\tilde \theta}_{i}^S$ can be rewritten as
	\begin{equation} \label{E26}
		{\tilde \theta}_{i}^S= \left(H_{fi}\sum_{k=1}^{N}\frac{1}{N}{e_i(k)z_{p,i}^\top(k)}{\hat \Sigma}_{i,N}^{-1/2}\right){\hat \Sigma}_{i,N}^{-1/2}.
	\end{equation}
	To bound the above term, we use a self-normalized martingale to deal with the leftmost term in the bracket. As for ${\hat \Sigma}_{i,N}^{-1/2}$, we use recent results from the smallest eigenvalue of the empirical covariance of causal Gaussian processes to bound it \cite{Ziemann2023note,Ziemann2023tutorial}, which establish the condition of PE. 
	
	\subsection{Persistence of Excitation} \label{Sct3.1}
	
	First, we make the following assumption regarding the selection of the past horizon $p$.
	\begin{assumption} \label{Asp2}
		The past horizon is chosen as $p=\beta\text{log}N$, where $\beta$ is large enough such that
		\begin{equation} \label{E27}
			\norm{CA_c^p}\norm{{\Sigma}_{x,N}} \leq N^{-3}. 
		\end{equation}
	\end{assumption}
	\begin{remark} \label{Rmk1}
		The above assumption ensures that the truncation bias $\Gamma_{fi}A_c^px_{k}$ is small enough, allowing the model \eqref{E15} to closely approximate an ARX model. To achieve this goal, the exponentially decaying term $A_c^p$ should counteract the magnitude of the state $x_{k}$. Since $A$ is marginally stable, $x_{k}$ scales at most polynomially with $k$ \cite{Ziemann2023tutorial}. Hence, the state norm $\norm{{\Sigma}_{x,N}}$ grows at most polynomially with $N$. Since $\rho(A_c)<1$, we have $\norm{A_c^p} = \mathcal{O}(\rho^p)$ for some $\rho > \rho(A_c)$. Taking $p=\beta\text{log}N$, we have $\norm{A_c^p} = \mathcal{O}(N^{-\beta/{\rm{log}(1/\rho)}})$. In this way, the condition \eqref{E27} will be satisfied for a large $\beta$.
	\end{remark}
	
	PE is equivalent to requiring that the empirical covariance ${\hat \Sigma}_{i,N}$ is positive definite. For this purpose, the number of samples $N$ should exceed a certain threshold, namely, burn-in time $N_{pe}$. 
	\begin{definition}\label{Def1}
		The burn-in time $N_{pe}$ is defined as
		\begin{equation} \label{E28}
			N_{pe}(\delta,\beta,i) \triangleq \text{min}\left\{N: N \geq N_0(N,\delta,\beta,i)\right\},
		\end{equation}
		where 
		\begin{equation} \label{E28a}
			\nonumber
			\begin{split}
				&N_0(N,\delta,\beta,i) 
				\triangleq c_0\tau_i\text{max}\left\{\sigma_e^2,1\right\}(\text{log}\frac{1}{\delta} + d_i\text{log}C_\text{sys}(N,\tau_i)),\\
				&C_\text{sys}(N,\tau_i) \triangleq \frac{N}{3\tau_i}\frac{{\norm{\Sigma_{i,N}}^2}}{{\lambda_{\text{min}}^2}(\Sigma_{i,\tau_i})},
				\tau_i = i+p = i +\beta\text{log}N, \\
			\end{split}
		\end{equation}
		$d_i = pn_y + \tau_i n_u$ and $c_0$ is a universal positive constant.
	\end{definition}
	\begin{remark} \label{Rmk2}
		To show that the above definition is not vacuous, we need to demonstrate that the condition $N\geq N_{0}(N,\delta,\beta,i)$ is feasible. For any given $\beta$,  $\tau_i$ increases logarithmically with $N$. Also, $\norm{\Sigma_{i,N}}^2$ grows polynomially with $N$, hence, the system theoretic term $\text{log}C_\text{sys}(N,\tau_i)$ increases at most logarithmically with $N$ \cite{Ziemann2023tutorial}. As a result, $N_{0}(N,\delta,\beta,i)$ grows polynomially with $N$.
	\end{remark}
	
	The condition of PE is summarized as follows:
	\begin{lemma} \label{Lem1}
		Fix a failure probability $0<\delta<1$, if $N\geq N_{pe}(\delta,\beta,i)$, then with probability at least $1-\delta$, we have
		\begin{equation} \label{E29}
			{\hat \Sigma}_{i,N} \succcurlyeq \frac{1}{16}\Sigma_{i,\tau_i} ,
		\end{equation}
		where $\lambda_{\text{min}}(\Sigma_{i,\tau_i}) > 0$.
	\end{lemma}
	
	\begin{proof}
		To conserve space, the complete proof is omitted here, which is identical to PE of the ARX model provided in Theorem 5.2 of \cite{Ziemann2023tutorial}.
	\end{proof}
	
	\begin{remark} \label{Rmk3}
		We emphasize that the above PE result is also helpful for analyzing the impact of weighting matrices $W_1$ and $W_2$ in the weighted SVD step \eqref{E13}, which will be investigated in detail in the future. To illustrate this, it has been shown in \cite{Qin2005novel} that $ \left(Z_p\Pi_{U_f}^{\perp}Z_p^\top\right)^{1/2}$ is approximately the optimal choice for $W_2$. A prerequisite is that $Z_p\Pi_{U_f}^{\perp}Z_p^\top$ should be positive definite, i.e., $\frac{1}{N}U_fU_f^\top \succ 0$ and $\frac{1}{N}Z_p\Pi_{U_f}^{\perp}Z_p^\top \succ 0$. According to the Schur complement, this is equivalent to 
		\begin{equation} \label{E30}
			\frac{1}{N}\begin{bmatrix}Z_pZ_p^\top & Z_pU_f^\top \\U_fZ_p^\top &U_fU_f^\top
			\end{bmatrix} ={\hat \Sigma}_{f,N}  \succ 0,
		\end{equation}
		which is essentially same as PE in Lemma \ref{Lem1} by taking $i=f$.
	\end{remark}
	
	\subsection{Stochastic Error } \label{Sct3.2}
	
	The bound of the stochastic error ${\tilde \theta}_{i}^S$ is based on the following three events:
	\begin{subequations} \label{E31}
		\begin{align}
			\mathcal{E}_{i,1} \triangleq &\left\{{\hat \Sigma}_{i,N}  \succcurlyeq \frac{1}{16} {\Sigma}_{i,\tau_i}  \right\},\\
			\mathcal{E}_{i,2} \triangleq &\left\{{\hat \Sigma}_{i,N}  \preccurlyeq \frac{3d_i}{\delta} {\Sigma}_{i,N}  \right\}, \\
			\nonumber
			\mathcal{E}_{i,3} \triangleq &\left\{ \norm{ \sum_{k=1}^{N}e_i(k)z_{p,i}(k)^\top \left(\Sigma + N {\hat \Sigma}_{i,N}\right)^{-1/2}}^2 \right. \leq\\ 
			&  \left.4\sigma_e^2{\text{log}\frac{\text{det}(\Sigma + N {\hat \Sigma}_{i,N})}{\text{det}\left(\Sigma\right)}} + 8\sigma_e^2\left(n_y\text{log}5+\text{log}\frac{3}{\delta} \right) \right\},
		\end{align}
	\end{subequations}
	where $\mathbb{P}(\mathcal{E}_{i,j}^c) \leq \delta/3$ for $j=1,2,3$. 
	The event $\mathcal{E}_{i,1}$ is due to PE in Lemma \ref{Lem1}, the event  $\mathcal{E}_{i,2}$ is derived from the matrix Markov inequality \cite{Ziemann2023tutorial}, and the event $\mathcal{E}_{i,3}$ is based on the result of self-normalized martingales \cite{Abbasi2011online}. The signal-to-noise ratio (SNR) is defined as 
	\begin{equation} \label{E32}
		{{\rm{SNR}}_{i,k}}  \triangleq \frac{\lambda_{\text{min}}({\Sigma}_{i,k})}{\sigma_e^2},
	\end{equation}
	and is assumed to be uniformly lower bounded for all possible $p$, which has been proven to be a quite general assumption\cite{Ziemann2023tutorial}. The bound of the stochastic error ${\tilde \theta}_{i}^S$ is summarized in the following lemma:
	\begin{lemma} \label{Lem2}
		Fix a failure probability $0<\delta<1$, if $N \geq N_{pe}(\delta/3,\beta,i)$, then with probability at least $1-\delta$, 
		\begin{equation} \label{E33}
			\norm{{\tilde \theta}_{i}^S}^2 \leq \frac{c\norm{H_{fi}}^2}{{{\rm{SNR}}_{i,\tau_i}}N} (d_i{\rm{log}}\frac{d_i}{\delta} + {\rm{log}} {\rm{det}}({\Sigma}_{i,N}{{\Sigma}_{i,\tau_i}^{-1}})),
		\end{equation}
		where $c$ is a universal constant which is independent of the system, $\delta$ and $\beta$.
	\end{lemma}
	\begin{proof}
		We sketch the proof here, which is similar to the proof of Theorem 5.1 in \cite{Ziemann2023tutorial}. The main idea is to prove that the event in Lemma \ref{Lem2} is subsumed by the intersection of three events $\mathcal{E}_{i,1}$, $\mathcal{E}_{i,2}$ and $\mathcal{E}_{i,3}$, with $\mathbb{P}(\mathcal{E}_{i,1}\cap\mathcal{E}_{i,2}\cap\mathcal{E}_{i,3}) \geq 1-\delta$, due to the fact that  $\mathbb{P}(\mathcal{E}_{i,1}^c\cup\mathcal{E}_{i,2}^c\cup\mathcal{E}_{i,3}^c) \leq \delta$. According to \eqref{E26}, we have 
		\begin{equation} \label{E34}
			\norm{{\tilde \theta}_{i}^S}^2 \leq  \underbrace{\norm{H_{fi}\sum_{k=1}^{N}\frac{1}{N}{e_i(k)z_{p,i}^\top(k)} {\hat \Sigma}_{i,N}^{-1/2}}^2}_{\text{noise term}} \underbrace{\norm{{\hat \Sigma}_{i,N}^{-1}}}_{\text{excitation term}}.
		\end{equation}
		The excitation term is bounded based on the event $\mathcal{E}_{i,1}$, i.e., 
		\begin{equation} \label{E35}
			\norm{{\hat \Sigma}_{i,N}^{-1}} \leq 16\lambda_{\text{min}}^{-1}({\Sigma}_{i,\tau_i}).
		\end{equation}
		For the noise term, we bound it by mimicking the form of the event $\mathcal{E}_{i,3}$. Due to $\mathcal{E}_{i,1}$, we have $2N{\hat \Sigma}_{i,N} \succcurlyeq N{\hat \Sigma}_{i,N} + \Sigma_i$, where $\Sigma_i = \frac{N}{16}{\Sigma}_{i,\tau_i}$. Hence, the noise term can be formulated as
		\begin{equation} \label{E36}
			\begin{split}
				&\norm{\frac{H_{fi}}{\sqrt{N}}\sum_{k=1}^{N}e_i(k)z_{p,i}(k)^\top({N\hat \Sigma}_{i,N})^{-1/2}}^2\leq\\
				&\frac{2\norm{H_{fi}}^2}{N}\norm{\sum_{k=1}^{N}e_i(k)z_{p,i}(k)^\top({\Sigma_i+N\hat \Sigma}_{i,N})^{-1/2}}^2.
			\end{split}
		\end{equation}
		Furthermore, according to $\mathcal{E}_{i,3}$ and $\mathcal{E}_{i,2}$, we have
		\begin{equation} \label{E37}
			\begin{split}
				&\norm{\sum_{k=1}^{N}e_i(k)z_{p,i}(k)^\top\left({\Sigma_i+N\hat \Sigma}_{i,N}\right)^{-1/2}}^2 \leq \\
				& 4\sigma_e^2{\text{log}\frac{\text{det}(\Sigma_i + N {\hat \Sigma}_{i,N})}{\text{det}\left(\Sigma_i\right)}} + 8\sigma_e^2\left(n_y\text{log}5+\text{log}\frac{3}{\delta}\right) \leq \\
				&4\sigma_e^2{\text{log}\text{det}(I + 48\frac{d_i}{\delta} {\Sigma}_{i,N}{\Sigma}_{i,\tau_i}^{-1})} + 8\sigma_e^2(n_y\text{log}5 + \text{log}\frac{3}{\delta}).\\
			\end{split}    	
		\end{equation}
		Combining \eqref{E36} and \eqref{E37} with \eqref{E35}, and simplifying them properly, the result \eqref{E33} is obtained.
	\end{proof}
	
	\subsection{Truncation Bias Term} \label{Sct3.3}
	
	As for the bias term ${\tilde \theta}_{i}^B$, we will see that it is dominated by the stochastic error ${\tilde \theta}_{i}^S$, given a proper selection of $p$ and the fact that the system is stable \cite{Ziemann2023tutorial}. The bias term is similarly decomposed as
	\begin{equation} \label{E38}
		{\tilde \theta}_{i}^B= \left(\Gamma_{fi}\bar A^p\sum_{k=1}^{N}{x_{k}z_{p,i}^\top(k)}(N\hat \Sigma_{i,N})^{-1/2}\right)(N\hat \Sigma_{i,N})^{-1/2}.
	\end{equation}
	Note that ${N\hat \Sigma}_{i,N} = \sum_{k=1}^{N}z_{p,i}(k)z_{p,i}^\top(k) \succcurlyeq z_{p,i}(k)z_{p,i}^\top(k)$, hence we have
	\begin{equation} \label{E39}
		\norm{z_{p,i}^\top(k)\left({N\hat \Sigma}_{i,N}\right)^{-1}z_{p,i}(k)} \leq 1.
	\end{equation}
	Using the triangle inequality,
	\begin{equation} \label{E40}
		\begin{split}
			&\norm{\sum_{k=1}^{N}{x_{k}z_{p,i}^\top(k)} \left({N\hat \Sigma}_{i,N}\right)^{-1/2}} \leq \\ 
			&\sum_{k=1}^{N}\norm{x_{k}} \norm{z_{p,i}^\top(k)\left({N\hat \Sigma}_{i,N}\right)^{-1/2}}.
		\end{split}	
	\end{equation}
	Then based on \eqref{E39} and the Cauchy-Schwarz inequality, we have
	\begin{equation} \label{E41}
		\norm{\sum_{k=1}^{N}{x_{k}z_{p,i}^\top(k)} \left({N\hat \Sigma}_{i,N}\right)^{-1/2}} \leq
		\sqrt{N\sum_{k=1}^{N}\norm{x_{k}}^2}.
	\end{equation}
	Now we introduce the following lemma:
	\begin{lemma} [Lemma E.5 in \cite{Ziemann2023tutorial}] \label{Lem3}
		Fix a failure probability $\delta$, there exists a universal constant $c$ such that with probability at least $1-\delta$, 
		\begin{equation} \label{E42}
			\sum_{k=1}^{N}\norm{x_{k}}^2 \leq c \sigma_e^2 n_x N \norm{\Sigma_{x,N}} {\rm{log}} \frac{1}{\delta}.
		\end{equation}
	\end{lemma}
	Using this lemma, inequality \eqref{E41} can be further simplified as 
	\begin{equation} \label{E43}
		\norm{\sum_{k=1}^{N}{x_{k}z_{p,i}^\top(k)} \left({N\hat \Sigma}_{i,N}\right)^{-1/2}} \leq \sqrt{c \sigma_e^2 n_x  \norm{\Sigma_{x,N}} N{\rm{log}} \frac{1}{\delta}}.
	\end{equation}
	Combining with the event $\mathcal{E}_{i,1}$, now we bound the bias as
	\begin{equation} \label{E44}
		\norm{{\tilde \theta}_{i}^B}^2 \leq \frac{16cn_x}{{{\rm{SNR}}_{i,\tau_i}}} \norm{\Gamma_{fi}\bar A^p}\norm{\Sigma_{x,N}}N{\rm{log}} \frac{1}{\delta}.
	\end{equation}
	According to Assumption \ref{Asp2}, by taking $p = \beta {\rm{log}} N$ such that
	\begin{equation} \label{E45}
		\norm{\Gamma_{fi}\bar A^p}\norm{\Sigma_{x,N}} \leq N^{-3},
	\end{equation}
	the bias error can be further bounded as
	\begin{equation} \label{E46}
		\norm{{\tilde \theta}_{i}^B}^2 \leq \frac{16cn_x}{N^2{{\rm{SNR}}_{i,\tau_i}}}{\rm{log}} \frac{1}{\delta}.
	\end{equation}
	
	As we can see, the truncation bias ${\tilde \theta}_{i}^B$ decays as $\mathcal{O}(1/N)$ and is dominated by the stochastic error ${\tilde \theta}_{i}^S$. By merging the two errors ${\tilde \theta}_{i}^S$ and ${\tilde \theta}_{i}^B$ together, and absorbing high order terms into the dominating term by inflating the constants accordingly, we obtain the following theorem.
	\begin{theorem} \label{The1}
		Fix a failure probability $\delta$, if $N \geq N_{pe}(\frac{\delta}{3},\beta,i)$, then with probability at least $1-2\delta$,
		\begin{equation} \label{E47}
			\norm{{\tilde \theta}_{i}}^2 \leq \frac{c\norm{H_{fi}}^2}{{{\rm{SNR}}_{i,\tau_i}}N} \left(d_i{\rm{log}}\frac{d_i}{\delta} + {\rm{log}} {\rm{det}}({\Sigma}_{i,N}{{\Sigma}_{i,\tau_i}^{-1}})\right),
		\end{equation}
		where $c$ is a universal constant which is independent of the system, $\delta$ and $\beta$.
	\end{theorem}
	
	\section{Robustness of Balanced Realization} \label{sec4}
	
	Following our roadmap, having obtained error bounds for each ARX model in Step 1, now we move to Step 2 and provide the overall bound for \eqref{E17} and error bounds for the balanced realization \eqref{E18}.
	
	\subsection{Overall Bound} \label{Sct4.1}
	
	First, we introduce the following lemma about the block matrix norm:
	\begin{lemma} [Lemma A.1 in \cite{Tsiamis2019finite}] \label{Lem4}
		Let $M$ be a block-column matrix defined as $M = \begin{bmatrix}
			M_1^\top&M_2^\top&\cdots &M_f^\top
		\end{bmatrix}^\top$, where all the $M_i$'s have the same dimension. Then, the block matrix $M$ satisfies
		\begin{equation} \label{E48}
			\norm{M} \leq \sqrt{f} \max\limits_{1\leq i\leq f} \norm{M_i}.
		\end{equation}
	\end{lemma}
	Based on this lemma, it is straightforward to obtain the following theorem regarding the error bound of the range space of the extended observability matrix $\widehat {\Gamma_{f}L_p}$.
	\begin{theorem} \label{The1}
		For a fixed probability $\delta$, if 
		\begin{equation} \label{E49}
			N \geq {\rm {max}} \left\{N_{pe}(\frac{\delta}{3f},\beta,i)\right\}, {\rm {for}} \ i=1,2,...f,
		\end{equation}
		then with probability at least $1-2\delta$, we have
		\begin{equation} \label{E50}
			\begin{split}
				&\norm{\widehat {\Gamma_{f}L_p} - {\Gamma_{f}L_p}} \leq \sqrt{f} \max\limits_{1\leq i\leq f} \norm{\tilde{\theta}_i} \leq \\
				&\sqrt{\frac{f}{N}} \max\limits_{1\leq i\leq f} \sqrt{\frac{c\norm{H_{fi}}}{{{\rm{SNR}}_{i,\tau_i}}} (d_i{\rm{log}}\frac{d_i}{\delta} + {\rm{log}} {\rm{det}}({\Sigma}_{i,N}{{\Sigma}_{i,\tau_i}^{-1}}))}. 
			\end{split}			
		\end{equation}
	\end{theorem}
	
	\subsection{Bounds on System Matrices} \label{Sct4.2}
	
	Assume that we know the true value of $\Gamma_{f}L_p$, and the SVD of $\Gamma_{f}L_p$ is 
	\begin{equation} \label{E51}
		\Gamma_{f}L_p = \begin{bmatrix}
			U_{1}&U_0
		\end{bmatrix}\begin{bmatrix}
			\Lambda_{1}&0\\0&0
		\end{bmatrix}\begin{bmatrix}
			V_{1}&V_0
		\end{bmatrix}^\top,
	\end{equation}
	so a balanced realization for $\Gamma_{f}$ and $L_p$ is 
	\begin{equation} \label{E52}
		\bar{\Gamma}_{f} = U_{1}\Lambda_{1}^{1/2}, \bar{L}_p = \Lambda_{1}^{1/2}V_{1}^\top.
	\end{equation}
	Moreover, the system matrices with respect to a similarity transform are obtained according to \eqref{E18}, which are denoted as $\left\{\bar{A},\bar{B},\bar{C},\bar{K}\right\}$. We now have the following result regarding the error bounds of the system matrices:
	\begin{theorem} \label{The2}
		If the following condition is satisfied:
		\begin{equation} \label{E53}
			\norm{\widehat {\Gamma_{f}L_p}- \Gamma_{f}L_p} \leq \frac{\sigma_n (\Gamma_{f}L_p)}{4},
		\end{equation}
		then there exists a unitary matrix $T$ such that, for a fixed probability $\delta$, if 
		\begin{equation} \label{E55}
			N \geq {\rm {max}} \left\{N_{pe}(\frac{\delta}{3f},\beta,i)\right\}, {\rm {for}} \ i=1,2,...f,
		\end{equation}
		then with probability at least $1-2\delta$, we have
		\begin{subequations} \label{E56}
			\begin{align}
				\nonumber
				&{\rm {max}}\left\{\norm{\hat{\Gamma}_{f}- \bar{\Gamma}_{f}T},\norm{\hat{L}_p- T^\top\bar{L}_p}\right\} \leq \\
				&2\sqrt{\frac{10n}{\sigma_n (\Gamma_{f}L_p)}} \norm{\widehat {\Gamma_{f}L_p}- \Gamma_{f}L_p},\\
				\nonumber
				&{\rm {max}}\left\{\norm{\hat{C}- \bar{C}T},\norm{\hat{K}- T^\top\bar{K}},\norm{\hat{B}- T^\top\bar{B}} \right\} \leq \\ 
				& \norm{\hat{\Gamma}_{f}- \bar{\Gamma}_{f}T}, \\
				&\norm{\hat{A}- T^\top \bar{A}T} \leq \frac{\sqrt{\norm{\Gamma_{f}L_p}}+\sigma_o}{\sigma_o^2}\norm{\hat{\Gamma}_{f}- \bar{\Gamma}_{f}T},
			\end{align}
		\end{subequations}
		where $\sigma_o = {\rm{min}}\left(\sigma_n(\hat{\Gamma}_f \underline{I}),\sigma_n(\bar{\Gamma}_f \underline{I})\right)$.
	\end{theorem}
	\begin{proof}
		The proof of the above SVD step is similar to \cite{Tsiamis2019finite,Oymak2019non}, thus, it is omitted here.
	\end{proof}
	Now we discuss the implications of our main results Theorems \ref{The1} and \ref{The2}. 
	\begin{enumerate}
		\item The dimension of the system affects PE and error bounds: As shown in \eqref{E28} and \eqref{E50}, the number of samples $N$ and the error bound scale with the dimension $d_i = p(n_u+n_y) + in_u$, which corresponds to the intuition that to estimate $\theta_i \in \mathbb{R}^{n_y \times d_i}$, we need at least as many number of independent equations coming from measurements.
		\item The error bounds scale linearly with ${\rm{log}}\frac{1}{\delta}$, which is consistent with the result in the asymptotic regime \cite{Ziemann2023tutorial}.
		\item According to Theorem \ref{The1}, the total error bound of $\norm{\widehat {\Gamma_{f}L_p}- \Gamma_{f}L_p}$ decays as $\mathcal{O}(1/\sqrt{N})$. This conclusion is drawn based on the following observations: The error bound scales linearly with the logarithm of  ${\rm{det}}({\Sigma}_{i,N}{{\Sigma}_{i,\tau_i}^{-1}})$ which increases at most polynomially with $N$, the SNR is uniformly lower bounded, $\norm{H_{fi}}$ is bounded and the future horizon $f$ is fixed.
		\item According to Theorem \ref{The2}, the error bounds of the system matrices $\left\{A,B,C,K\right\}$ depend linearly on the error bound of $\norm{\widehat {\Gamma_{f}L_p}- \Gamma_{f}L_p}$, thus, they also decay as $\mathcal{O}(1/\sqrt{N})$. This indicates that as long as the identification error of $\Gamma_{f}L_p$ is small, the results from the subsequent SVD step are robust.
	\end{enumerate}
	\section{Discussion and Future Work} \label{sec5}
	
	This paper presents a finite sample analysis of a parsimonious SIM that estimates a bank of ARX models using OLS in parallel. It reveals that the convergence rates for estimating Markov parameters and system matrices are  $\mathcal{O}(1/\sqrt{N})$, consistent with classical asymptotic results. Besides, PE and the role of past horizon are also discussed. Although the algorithm studied in this paper streamlines the weighted SVD and realization steps, we believe that the findings herein pave the way for comprehensively grasping the broader landscape of the family of SIMs. Here are a few aspects we would like to discuss:
	\begin{enumerate}
		\item Our bound is not tight: Just like the existing bounds for partially observed systems, our bound is not tight, either. In this paper, PE in the $f$ ARX models are dealt with separately, which is convenient but somewhat conservative, given that the past output and input $Z_p$ is reused for every estimation. Our bound can be further optimized; however, akin to the challenge of finding an efficient SIM in the asymptotic regime, the pursuit of a lower bound in the non-saymptotic regime for partially observed systems is more challenging.
		\item Selection of past and future horizons: A similar suggestion as other work is given on the selection of the past horizon $p$, however, the guidance for the future horizon $f$ is not given. While it is evident that $f$ influences the error bound, a comprehensive analysis of its impact is not given. Using asymptotic theory \cite{Bauer2002some}, it has been shown that the variance of the estimates of the CVA method improves as $f$ increases, thus, it is interesting to study the impact of $f$ in the non-asymptotic regime.
		\item Comparison between different weighting matrices: We consider a simplified SIM, i.e., taking the weighting matrices $W_1 = W_2 = I$. It is generally believed that the optimal choices usually depend on specific data sets, like the CVA method. In the future, we will analyze the influence of different weighting matrices on the realization step. As we remark in the persistence of excitation, our work provides some preliminary insights on the validity of existing weighting matrices.
	\end{enumerate}
	

	\bibliographystyle{IEEEtran}
	\bibliography{refs}

\end{document}